\documentclass[runningheads]{llncs}  
\usepackage[utf8]{inputenc}
\usepackage{amsmath,amssymb,color,graphicx}
\usepackage[usenames,dvipsnames,svgnames,table]{xcolor}
\usepackage{blkarray}
\usepackage{verbatim}
\usepackage{cite}
\usepackage{tcolorbox}
\usepackage{algorithm}
\usepackage[noend]{algpseudocode}
\usepackage{hyperref}
\usepackage{xspace}
\usepackage{array}

\let\doendproof\endproof
\renewcommand\endproof{~\hfill$\qed$\doendproof}

\newcommand{\remove}[1]{{}}

\usepackage[textsize=tiny]{todonotes}

\newcommand{\arraycontrol}{\footnotesize}
\newcommand{\Stick}{{Stick}\xspace}
\newcommand{\StickAB}{{Stick$_{AB}$}\xspace} 
\newcommand{\StickA}{{Stick$_A$}\xspace} 
\newcommand{\StickfA}{{Stick$_{\forall A}$}\xspace} 

\usepackage{wrapfig}

\title{Recognition and Drawing of Stick Graphs}
\author{ Felice De Luca\inst{1}  \and Md Iqbal Hossain\inst{1} \and   Stephen Kobourov\inst{1}\and Anna Lubiw\inst{2}\and Debajyoti Mondal\inst{3}}

\institute{
 Department of Computer Science, University of Arizona, USA \email{\{felicedeluca,hossain\}@email.arizona.edu, kobourov@cs.arizona.edu}
\and 
Cheriton School of Computer Science, University of Waterloo, Canada \email{alubiw@uwaterloo.ca}
\and 
Department of Computer Science, University of Saskatchewan, Canada \email{dmondal@cs.usask.ca}} 

\begin{document}
\maketitle

\begin{abstract}
A \emph{Stick graph} is an intersection graph of axis-aligned segments such that the left end-points of the horizontal segments and the bottom end-points of the vertical segments lie on a ``ground line,'' a line with slope $-1$. It is an open question to decide in polynomial time whether a given bipartite graph $G$ with bipartition $A\cup B$ has a \Stick representation where the vertices in $A$ and $B$ correspond to horizontal and vertical segments, respectively. We prove that $G$ has a \Stick representation if and only if there are orderings of $A$ and $B$ such that $G$'s bipartite adjacency matrix with rows $A$ and columns $B$ excludes three small `forbidden' submatrices. This is similar to characterizations for other classes of bipartite intersection graphs.  

We present an algorithm to test whether given orderings of $A$ and $B$ permit a \Stick representation respecting those orderings, and to find such a representation if it exists.  The algorithm runs in time linear in the size of the adjacency matrix. For the case when only the ordering of $A$ is given, we present an $O(|A|^3|B|^3)$-time algorithm.
When neither ordering is given, we present some partial results about graphs that are, or are not, \Stick representable. 
\end{abstract}
 
\section{Introduction}
Let $\mathcal{O}$ be a set of geometric objects in the Euclidean plane. The \emph{intersection graph} of $\mathcal{O}$ is a graph where each vertex 
corresponds to a distinct object in $\mathcal{O}$, and two vertices are adjacent if and only if the corresponding objects intersect. Recognition of intersection graphs that arise from different types of geometric objects such as segments, rectangles,  discs, intervals, etc., is a classic problem in combinatorial geometry. 
Some of these classes, such as interval graphs~\cite{booth1976testing},  can be recognized in polynomial-time, whereas many others are NP-hard~\cite{DBLP:conf/wg/CardinalFMTV17,kratochvil1994special,mustata2013unit}. 
There are many beautiful results that characterize intersection classes in terms of a vertex ordering without certain forbidden patterns, and recently, Hell {\em et al.}~\cite{hell2014ordering} unified many previous results by giving a general polynomial time recognition algorithm for all cases of small forbidden patterns. 

In this paper we study a class of bipartite intersection graphs called \emph{Stick graphs}.  
A \emph{Stick graph} is an intersection graph of axis-aligned segments with the property that the left end-points of horizontal segments and the bottom end-points of vertical segments all lie on a \emph{ground line}, $\ell$, which we take, without loss of generality, to be a line of slope $-1$.  See Fig.~\ref{fig:intro}(a)--(b).
It is an open problem to recognize \Stick graphs in polynomial time~\cite{chaplick2015grid}. 

\begin{figure}[pt]
\centering
\includegraphics[width=\textwidth]{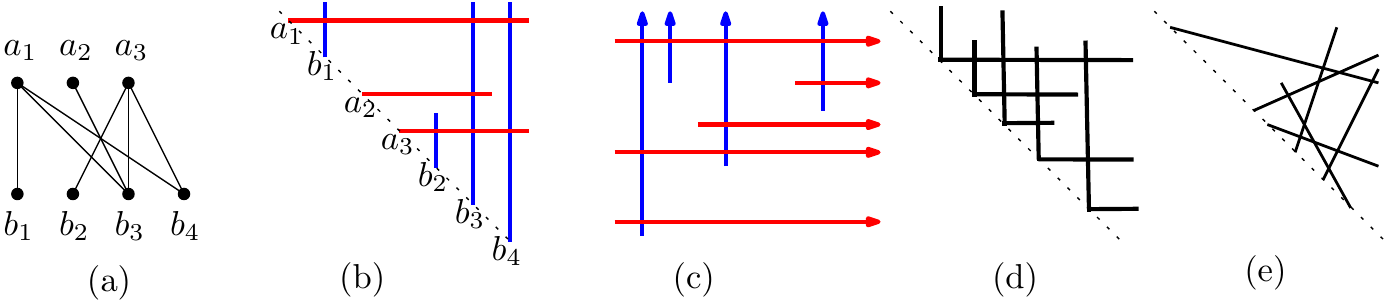} 
\caption{(a) A bipartite graph $G=(A\cup B,E)$. (b) A \Stick representation of $G$. (c)--(e) Illustration for different types of intersection representations.  
 (c) A \emph{2DOR} representation. (d) A \emph{Hook} representation. (e) A grounded segment representation. 
 }
\label{fig:intro}
\end{figure}

\Stick graphs lie between two well-studied classes of bipartite intersection graphs.
First of all, they are a subset of the \emph{grid intersection graphs} (GIG)~\cite{hartman1991grid}---intersection graphs of horizontal and vertical segments in the plane---which are NP-complete to recognize~\cite{kratochvil1994special}.
When all the horizontal segments extend rightward to infinity and the vertical segments extend upward to infinity, we obtain the subclass of \emph{2-directional orthogonal ray} (2DOR) graphs (e.g., see Fig.~\ref{fig:intro}(c)), which can be recognized in polynomial time~\cite{shrestha2010orthogonal}.
It is easy to show that 
every 2DOR graph is a \Stick 
graph---truncate each ray at a ground line placed above and to the right of every intersection point (and then flip the picture upside-down). 
Thus the class of \Stick graphs lies strictly between these two classes.

What the two classes (GIG and 2DOR) have in common is a nice characterization in terms of vertex orderings.
A bipartite graph $G$ with vertex bipartition $A \cup B$ can be represented as a \emph{bipartite adjacency matrix}, $M(G)$ with rows and columns corresponding to $A$ and
$B$, respectively, and a 1 in row $i$, column $j$, if $(i,j)$ is an edge. 
Both GIG graphs and 2DOR graphs can be characterized as graphs $G$ for which $M(G)$ has a row and column ordering without certain `forbidden' submatrices. (Details below.)  
Many other bipartite intersection graphs can be similarly characterized in terms of forbidden submatrices, see~\cite{klinz1995permuting}.

One of our main results is a similar characterization of \Stick graphs.
Specifically, we will prove that a bipartite graph $G$ with vertex bipartition $A \cup B$ has a  \Stick representation with vertices of $A$ corresponding to  horizontal segments and vertices of $B$ corresponding to the vertical segments if and only if there is an ordering of $A$ and an ordering of $B$ such that $M(G)$ has no submatrix of the following form, where $*$ stands for either 0 or 1:
\begin{center}
\mbox{\small \begin{blockarray}{[ccc]}
   $*$  & $1$ & $*$ \\
   $*$  & $0$ & $1$ \\
   $1$  & $*$ & $*$ \\
\end{blockarray}} \ \ \ \ \  %
{\arraycontrol \begin{blockarray}{[cc]}
   $1$ & $*$ \\
   $0$ & $1$ \\
   $1$ & $*$ \\
\end{blockarray}}\ \ \ \ \  %
{\arraycontrol \begin{blockarray}{[ccc]}
   $*$ & $1$ & $*$ \\
   $1$ & $0$ & $1$ \\
\end{blockarray}}  
\end{center}

Although this characterization does not (yet) give us a polynomial time algorithm to recognize \Stick graphs, it allows us to make some progress.
Given a bipartite graph $G$ with vertex bipartition $A \cup B$, we want to know if $G$ has a \Stick representation 
with $A$ and $B$ corresponding to horizontal and vertical segments, respectively.
It is easy to show that a solution
to this problem is 
completely determined by a total ordering $\sigma$ of the vertices of $G$ 
corresponding to the order (from left to right) in which the segments touch the ground line.
A natural way to tackle the recognition of \Stick graphs is as a hierarchy of problems, each (possibly) more difficult than the next:
\begin{enumerate}
\item[(i)] {\bf Fixed $A$s and $B$s}: In this case an ordering, $\sigma_a$, of the vertices in $A$ and an ordering, $\sigma_B$, of the vertices in $B$ 
are given, and the output ordering $\sigma$ must respect these given orderings.  Because of our forbidden submatrix characterization, this problem can be solved in polynomial time.
\item[(ii)] {\bf Fixed $A$s}: 
In this case only  the ordering $\sigma_A$ is given. 
\item[(iii)] {\bf General Stick graphs}:
In this case, neither $\sigma_A$ nor $\sigma_B$ is given, i.e.,  there is no restriction on the ordering of the vertices.
\end{enumerate}

\medskip\noindent{\bf Our Results:}
We give an algorithm with run-time $O(|A||B|)$ for problem (i).  This is faster than naively looking for the forbidden submatrices.  (And in fact, we use our algorithm  to prove the forbidden submatrix characterization).  Furthermore, the algorithm will find a \Stick representation when one exists.

We give an algorithm for problem (ii) with run time $O(|A|^3 |B|^3)$ that uses the forbidden submatrix characterization and reduces the problem to 2-Satisfiability. For problem (iii), recognizing \Stick graphs, we give some conditions that ensure a graph is a  \Stick graph, and some conditions that ensure a graph is not a \Stick graph.

\medskip\noindent{\bf Related Work:}
We now review the research related to the recognition of intersection graphs, in particular those that are bipartite.

 Interval graphs, i.e., intersection graph of horizontal intervals on the real line, can be recognized in linear time~\cite{booth1976testing, corneil2009lbfs}. 
 Bipartite interval graphs with a fixed bipartition are known as  \emph{interval bigraphs} (IBG)~\cite{muller1997recognizing, das1989interval}, and can be recognized in polynomial time~\cite{muller1997recognizing}. 
 In contrast to the interval graphs, no linear-time recognition algorithm is known for IBG.
 
Many bipartite graph classes have been  characterized in terms of forbidden submatrices of the graph's bipartite adjacency matrix, and a rich body of research examines when the rows and columns of a matrix can be permuted to avoid forbidden submatrices
~\cite{klinz1995permuting}.  For example, a graph $G$ is chordal bipartite if and only if $M(G)$ can be permuted to avoid the matrix $\gamma_1$ in Fig.~\ref{fig:matrices}~\cite{klinz1995permuting}, which led to a polynomial-time algorithm~\cite{DBLP:journals/siamcomp/Lubiw87}.   $G$ is a bipartite permutation graph if and only if  $M(G)$ can be permuted to avoid  $\gamma_1, \gamma_2$, and $\gamma_3$~\cite{DBLP:journals/networks/ChenY93}.

A graph is a \emph{two-directional orthogonal ray} (2DOR) graph if it admits an intersection representation of upward and rightward rays~\cite{soto2011jump, shrestha2010orthogonal}.   A graph is a 2DOR graph if and only if its incidence matrix admits a permutation of its rows and columns that avoids  $\gamma_1$ and $\gamma_2$~\cite{shrestha2010orthogonal}.

There is a linear-time algorithm to recognize 2DOR graphs~\cite{cogis1982ferrers, shrestha2010orthogonal}.
If there are 3 or 4 allowed directions for the rays, 
 then the graphs are called 3DOR or 4DOR graphs,  respectively. 
 Felsner {\em et al.}~\cite{DBLP:conf/mfcs/FelsnerMM13} showed that  if the direction (right, left, up, or down) for each vertex is given, then the existence of a 4DOR representation respecting the given directions can be decided in polynomial time. If the horizontal elements are segments and the vertical elements are rays, then the corresponding intersection graphs are called \emph{SegRay} graphs~\cite{chan2014exact, chaplick2015grid, chaplick2013stabbing,kalz2005orthogonal}. A graph $G$ is a SegRay graph if and only if $M(G)$ can be permuted to avoid $\gamma_4$~\cite{chaplick2014intersection}.

 \begin{figure}[tbp]
    \centering
$$    \gamma_1=\begin{bmatrix}
      1 & 0  \\
      1 & 1  \\
   \end{bmatrix}\ \ \  
   \gamma_2=\begin{bmatrix}
     1 & 0  \\
     0 & 1  \\
   \end{bmatrix}\ \ \ 
   \gamma_3=\begin{bmatrix}
     1 & 1  \\
     0 & 1  \\
   \end{bmatrix}\ \ \ 
   \gamma_4=\begin{bmatrix}
     1 & 0  & 1 \\
     $*$ & 1  & $*$ \\
   \end{bmatrix}\ \ \ 
   \gamma_5=\begin{bmatrix}
  $*$ & 1 & $*$ \\
  1 & 0 & 1 \\
  $*$ & 1 & $*$ \\
   \end{bmatrix}\ \ \ 
$$
    \caption{Forbidden submatrices, where $*$ stands for either 0 or 1.}
    \label{fig:matrices}
\end{figure}

The time-complexity questions for  3DOR, 4DOR and SegRay are all open.

The class of \emph{segment graphs} contains the graphs that can be represented as intersections of segments (with arbitrary slopes and  intersection angles).
 Every planar graph has a segment intersection representation~\cite{chalopin2009every}.
 Restricting to axis-aligned segments
 gives rise to \emph{grid intersection graphs} (GIG)~\cite{kratochvil1994special}.
 A bipartite graph is a GIG graph 
 if and only if its incidence matrix admits a permutation of its rows and columns that avoids
 $\gamma_5$~\cite{hartman1991grid}.
 If all the segments must have the same length, then the graphs are known as \emph{unit grid intersection graphs} (UGIG)~\cite{mustata2013unit}. The recognition problem is NP-complete for both GIG~\cite{kratochvil1994special} and UGIG~\cite{mustata2013unit}.
 We note that 4DOR is a subset of UGIG but Stick is not~\cite{chaplick2015grid}.

Researchers have examined further restrictions on GIG. For example, the graphs that admit a GIG representation with the additional constraint that all the segments must intersect (or be ``stabbed by'') a ground line form the \emph{stabbable grid intersection} (StabGIG) graph class~\cite{chaplick2015grid}. 

Another class of intersection graphs that restricts the objects on a ground line is 
defined in terms of \emph{hooks}.
A \emph{hook} consists of a center point on the ground line together with an incident vertical segment and horizontal segment above the ground line.
\emph{Hook} graphs are intersection graphs of hooks
~\cite{catanzaro2017max, hixon2013hook, soto2015pbox}, e.g., see Fig.~\ref{fig:intro}(d). 
Hook graphs are also known as \emph{max point-tolerance graphs}~\cite{catanzaro2017max} and \emph{heterozygosity graphs}~\cite{halldorsson2010clark}. The bipartite graphs that admit a Hook representation are called BipHook~\cite{chaplick2015grid}. 
 The complexities of recognizing the classes
StabGIG, BipHook, and \Stick are all  open~\cite{chaplick2015grid}. Chaplick {\em et al.}~\cite{chaplick2015grid} examined the containment relations of these graph classes.

\emph{Grounded segment representations} are a generalization of \Stick representations, where the segments can have arbitrary slopes, e.g., see Fig.~\ref{fig:intro}(e). Note that the segments are still restricted to lie on the same side of the ground line.  Cardinal {\em et al.}~\cite{DBLP:conf/wg/CardinalFMTV17} showed that the problem of deciding whether a graph admits a grounded segment representation is $\exists\mathbb{R}$-complete. We refer to~\cite{DBLP:journals/combinatorics/CabelloJ17,DBLP:conf/wg/CardinalFMTV17} for other related classes such as outersegment and  outerstring graphs, and for the study of their containment relations.

 The following table summarizes the time complexities of recognizing different classes of bipartite intersection graphs, where $n$ and $m$ are the sizes of the two vertex sets of the bipartition. 

\newcolumntype{P}[1]{>{\centering\arraybackslash}p{#1}}
\begin{table}[h]
\begin{minipage}{\textwidth}
\centering
\begin{tabular}{ | >{\centering\arraybackslash}m{7.6cm}|P{3cm}|P{1.2cm}|  }
 \hline
 \textbf{Graph Class} & \textbf{Time Complexity} & \textbf{Ref}\\
 \hline
  \raisebox{-1.5ex}{Chordal Bipartite Graphs} & $O((n+m)^2)$, or $|E| \log (n+m)$ & \raisebox{-1.5ex}{\cite{DBLP:journals/siamcomp/Lubiw87,DBLP:journals/ipl/Spinrad93}}\\\hline
Bipartite Permutation Graphs & $O(nm)$-time & \cite{DBLP:journals/dam/SpinradBS87} \\\hline
 2-Directional Ray Graphs (2DOR) & $O(nm)$-time
 & \cite{cogis1982ferrers, shrestha2010orthogonal} \\ \hline
3- or 4-Directional Ray Graphs (3DOR, 4DOR) & Open & \cite{cogis1982ferrers, shrestha2010orthogonal} \\\hline
4-DOR with given directions for vertices & $f(n,m)$-time\footnote{Multiplication time for two $(n+m)\times (n+m)$ matrices} & \cite{DBLP:conf/mfcs/FelsnerMM13} \\ \hline
3-DOR with a given bipartition $(A\cup B)$, and an ordering for $A$s, i.e., vertical rays &  $O((n+m)^2)$-time & \cite{DBLP:conf/mfcs/FelsnerMM13} \\ \hline
Grid Intersection Graphs (GIG) & NP-complete & \cite{kratochvil1994special}  \\ \hline
Unit Grid Intersection Graphs (UGIG) & NP-complete &   \cite{mustata2013unit} \\ \hline
Grounded Segment Intersection Graphs & $\exists\mathbb{R}$-complete &   \cite{DBLP:conf/wg/CardinalFMTV17} \\ \hline
StabGIG, SegRay,  Hook, BipHook and \Stick Graphs  & Open & \cite{chaplick2015grid,UdoHoffmann}\\ \hline 

 \hline
\end{tabular}
\end{minipage}
\end{table}

\section{Fixed $A$s and $B$s} 
\label{sec:AB-fixed}
In this section we study  \Stick representations of graphs with a fixed bipartition of the vertices and fixed vertex orderings for each vertex set. 
We call this problem \StickAB, defined formally as follows.

\begin{tcolorbox}[colframe=gray,colback=lightgray!20!white,boxrule=1pt,arc=0.4em,boxsep=-1mm]
 \textbf{Problem:} {\textsc{\Stick Representation with Fixed $A$s and $B$s (\StickAB)}}\\
 \textbf{Input:} A bipartite graph $G=(A{\cup}B,E)$, an ordering $\sigma_A$ of the vertices in $A$, and an ordering $\sigma_B$ of the vertices in $B$.\\
 \textbf{Question:}
 Does $G$ admit a \Stick representation 
 such that the $i$th horizontal segment on the ground line $\ell$ corresponds to the $i$th vertex of $\sigma_A$ and the $j$th vertical segment on $\ell$ corresponds to the $j$th vertex of $\sigma_B$? 
\end{tcolorbox}

 %
We first present an $O(|A||B|)$-time algorithm for \StickAB. 
A \Stick representation is totally determined by the order $\sigma$ of the segments' intersection with the ground line (details in the proof of Lemma~\ref{lem:simple}). 
Thus the idea of the algorithm is to impose some ordering constraints between the vertices of $A$ and $B$ 
based on some submatrices of the adjacency matrix of $G$. We show that the required \Stick representation exists if and only if there exists a total order $\sigma$ of $(A \cup B)$ that satisfies the constraints and preserves the given orderings $\sigma_A$ and $\sigma_B$. We now describe the details.

Assume that $\sigma_A = (a_1,\ldots, a_n)$ and  $\sigma_B = (b_1,\ldots, b_m)$. Let $M$ be the ordered bipartite adjacency matrix of $A$ and $B$, i.e., $M$ has rows $a_1,\ldots,a_n$ and columns $b_1,\ldots,b_m$,  where the entry $m_{i,p}$, i.e.,  the entry at the $i$th row and $p$th column, is 1 or 0 depending on whether $a_i$ and $b_p$ are adjacent or not, as illustrated in  Fig.~\ref{fig:simple}(a). 
 
We start with the constraints $a_{i-1} \prec a_i$, where  $2\le i\le n$, and $b_{p-1} \prec b_p$, where $ 2\le p\le m$ to enforce the given orderings $\sigma_A$ and  $\sigma_B$. We now add some more constraints, as follows.

\begin{enumerate}
\item[]\textbf{$C_1$:} If an entry $m_{i,p}$ is 1, then add the constraint $a_i \prec b_p$, e.g., see the black edges in  Fig.~\ref{fig:simple}(b). 

\item[]\textbf{$C_2$:} If $M$ contains an ordered submatrix 
\mbox{\scriptsize \begin{blockarray}{ccc}
& $b_p$ & $b_q$ \\
\begin{block}{c[cc]}
  {$a_i$\ } & 1 & $*$ \\
  {$a_j$\ } & 0 & 1 \\
\end{block}
\end{blockarray}} , %
then add the constraint $b_p \prec a_j$.   For example, see the gray edges in  Fig.~\ref{fig:simple}(b).
\end{enumerate}

We now test whether the set of constraints is consistent. Consider a directed graph $H$ with vertex set $(A \cup B)$, where each constraint corresponds to a directed edge (Fig.~\ref{fig:simple}(b)). Then the set of constraints is consistent if and only if $H$ is acyclic, and the following lemma claims that this occurs if and only if the graph admits a \Stick representation.

\begin{figure}[pt]
\centering
\includegraphics[width=.9\textwidth]{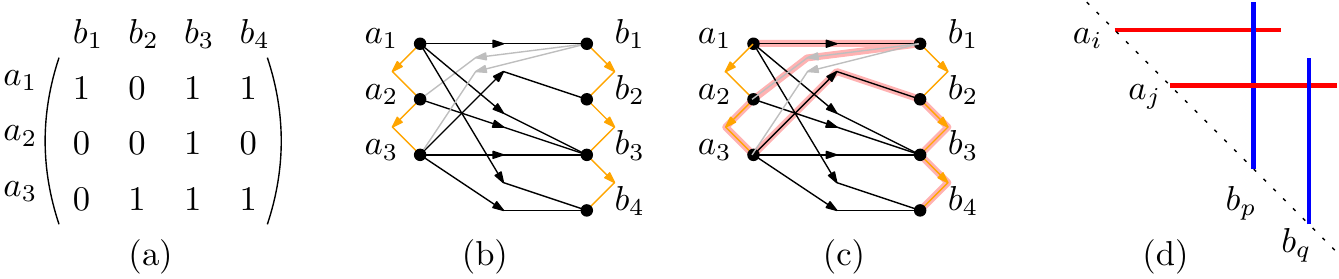} 
\caption{(a) The incidence matrix $M$ for the graph of Fig.~\ref{fig:intro}(a). 
 (b) The directed graph $H$. (c) A total order $(a_1,b_1,a_2,a_3,b_2,b_3,b_4)$ of the vertices in $H$. The corresponding \Stick drawing is in Fig.~\ref{fig:intro}(b). (d) Illustration of a forbidden ordering if $m_{j,p}=0$. }
\label{fig:simple}
\end{figure}

\begin{lemma}
\label{lem:simple}
$G$ admits a \Stick representation 
respecting $\sigma_A$ and $\sigma_B$ if and only if $H$ is acyclic, i.e., the constraints are consistent.
\end{lemma}
\begin{proof}
We first show that the constraints are necessary. 
Every constraint between two vertices of the same set is implied by $\sigma_A$ or  $\sigma_B$. 
For Condition $C_1$, observe that a horizontal segment $a_i$ can intersect a vertical segment $b_p$ only if $a_i$ precedes $b_p$, i.e., we must have $a_i \prec b_p$. For Condition $C_2$, we already have $a_i\prec a_j, b_p\prec b_q, a_i\prec b_p, a_j\prec b_q$.  If we assume that $a_j\prec b_p$,  then we have $a_i \prec a_j\prec b_p \prec b_q$, and to reach the vertical segment $b_q$,  $a_j$  would intersect $b_p$.  Since $m_{j,p} = 0$,  this intersection is forbidden. Fig.~\ref{fig:simple}(d) illustrates this scenario. Therefore,  we must have the constraint $b_p \prec a_j$.
Since all the constraints are necessary, if $G$ admits an intersection representation,  then the set of constraints is consistent.

We now prove the converse.  Suppose the set of constraints is consistent.  Take a total order of $A \cup B$ which is consistent with all the constraints, e.g., see  Fig.~\ref{fig:intro}(c).
This is a ``topological order'' of $H$.
Initiate the drawing of the corresponding orthogonal segments in this order on the ground line $\ell$. 
This determines the $y$-coordinate of every $a \in A$ and the $x$-coordinate of every $b \in B$.
For each vertex $a\in A$, let $\max_B(a)$ be the neighbor of $a$ in $G$ with the largest index. We extend the horizontal segment corresponding to $a$ to the right until the  $x$-coordinate of $\max_B(a)$. 
Similarly, for each vertex $b \in B$, let $\min_A(b)$ be the neighbor of $b$ in $G$ with the minimum index.  We extend the vertical segment corresponding to $b$ upward until the $y$-coordinate of $\min_A(b)$. 

We must show that the resulting drawing does not contain any forbidden intersection.
Suppose by contradiction that the  segments of $a_j$ and $b_p$ intersect, but they are not adjacent in $G$, i.e., $m_{j,p}=0$. We now have $a_j \prec b_p$, and the entries $b_q=\max_B(a_j)$ and $a_i=\min_A(b_p)$ give the submatrix 
described in Condition $C_2$, thus the constraint $b_p \prec a_j$ applies,  a contradiction. 
\end{proof}

An algorithm to solve \StickAB follows immediately, and can be implemented in linear time in the size of the adjacency matrix $M$.
\begin{theorem}\label{thm:fafb}
There is an $O(|A||B|)$-time algorithm to decide the \StickAB problem, and construct a \Stick representation if one exists.
\end{theorem}

\begin{proof}
The algorithm was given above: We construct the directed graph $H$ from the 0-1 matrix $M$ and test if $H$ is acyclic.  This correctly decides \StickAB by Lemma~\ref{lem:simple}.  Furthermore, if $H$ is acyclic, then we can construct a \Stick representation as specified in the proof of Lemma~\ref{lem:simple}. Pseudocode for the Algorithm~\ref{algo1} is given in Appendix~\ref{app:algo}. 

The matrix $M$ has 
size $O(nm)$ where $n=|A|$ and $m=|B|$, and the graph  $H$ has $n+m$ vertices and $O(nm)$ edges. 
We can test acyclicity of a graph and find a topological ordering in linear time.  
Also, the construction of the \Stick representation is clearly doable in linear time.  

Thus we only need to give details on constructing $H$ in time $O(nm)$.
We can construct the edges of $H$ that correspond to $\sigma_A$ and $\sigma_B$ in time $O(n+m)$.  The edges arising from constraints $C_1$ correspond to the 1's in the matrix $M$, so we can construct them in $O(nm)$ time.
The edges arising from constraints $C_2$ correspond to some of the 0's in the matrix $M$.  Specifically, a 0 in position $m_{j,p}$ gives a $C_2$ constraint $b_p \prec a_j$
if and only if there is a 1 in row $j$ to the right of the 0 and a 1 in column $p$ above the 0.  
We can flag the 0's that have a 1 to their right by scanning each row of $M$ from right to left.  Similarly, we can flag the 0's that have a 1 above them by scanning each column of $M$ from bottom to top.  These scans take time $O(nm)$.  Finally, if a 0 in $M$ has both flags, then we add the corresponding edge to $H$.
The total time is $O(nm)$.
\end{proof}

Lemma~\ref{lem:simple} also yields 
a forbidden submatrix characterization for \StickAB.

\begin{theorem}\label{thm:key}
An instance of \StickAB with graph $G=(A \cup B, E)$ has a solution if and only if $G$'s ordered adjacency matrix $M$ has no ordered submatrix of the following form:

\begin{center}
$P_1 =$ {\arraycontrol \begin{blockarray}{cccc}
& $b_p$ & $b_q$ & $b_r$\\
\begin{block}{c[ccc]}
  {$a_i$\ } & $*$  & $1$ & $*$ \\
  {$a_j$\ } & $*$  & $0$ & $1$ \\
  {$a_k$\ } & $1$  & $*$ & $*$ \\
\end{block}
\end{blockarray}} , %
$P_2=${\arraycontrol \begin{blockarray}{ccc}
& $b_p$ & $b_q$ \\
\begin{block}{c[cc]}
  {$a_i$\ } & $1$ & $*$ \\
  {$a_j$\ } & $0$ & $1$ \\
  {$a_k$\ } & $1$ & $*$ \\
\end{block}
\end{blockarray}} , %
$P_3=${\arraycontrol \begin{blockarray}{cccc}
& $b_p$ & $b_q$ & $b_r$\\
\begin{block}{c[ccc]}
  {$a_i$\ } & $*$ & $1$ & $*$ \\
  {$a_j$\ } & $1$ & $0$ & $1$ \\
\end{block}
\end{blockarray}} . 
\end{center}
\end{theorem}
\noindent
Observe that $P_2$ and $P_3$ are special cases of $P_1$ with $p{=}q$ and $j{=}k$, respectively.

\begin{proof}
We will use the graph $H$  that we  constructed above and used in Lemma~\ref{lem:simple}. 
By Lemma~\ref{lem:simple}, the theorem statement is equivalent to the statement that $M$ has a submatrix $P_1, P_2$ or $P_3$ if and only if $H$ has a directed cycle.

We first show that if the matrix $M$ has one of the ordered submatrices $P_1,P_2,P_3$ then $H$ has a directed cycle.  
For $P_1$, the cycle in $H$ is $a_k\prec b_p $ (by $C_1$), $b_p\prec b_q$ (by $\sigma_B$), $b_q\prec a_j$ (by $C_2$), $a_j\prec a_k$ (by $\sigma_A$). 
For $P_2$, the cycle is $b_p\prec a_j $ (by $C_2$), $a_j\prec a_k$ (by $\sigma_A$), $a_k\prec b_p$ (by $C_1$). 
For $P_3$, the cycle is $b_q\prec a_j $ (by $C_2$), $a_j\prec b_p$ (by $C_1$), $b_p\prec b_q$ (by $\sigma_B$).

To prove the other direction, suppose that $H$ has a directed cycle $O$.  We will show that $M$ has one of the submatrices 
$P_1,P_2,P_3$. 
Let $b_q$ be the rightmost vertex of $O$ in $\sigma_B$, and let $(b_q, z)$ be the outgoing edge of $b_q$ in $O$.   
Since $b_q$ is the rightmost vertex of $O$ in $\sigma_B$, $z$ must be a vertex $a_j$ of $A$. The constraint $b_q\prec a_j$ can only be added by $C_2$. Therefore, we must have the configuration 
\mbox{\scriptsize \begin{blockarray}{ccc}
& $b_q$ & $b_r$ \\
\begin{block}{c[cc]}
  {$a_i$\ } & 1 & * \\
  {$a_j$\ } & 0 & 1 \\
\end{block}
\end{blockarray}} . %
The path   can now continue from $a_j$ following zero or more $A$ vertices, but to complete the cycle, it eventually needs to reach a vertex $b_p$ of $B$. Since $b_q$ is the rightmost in $\sigma_B$, $b_p$ must appear either  before $b_q$ or coincide with $b_q$. 
First suppose that $b_p\not = b_q$. If the outgoing edge of $a_j$ is $(a_j,b_p)$, then we obtain the configuration $P_3$. Otherwise, the path visits several vertices of $A$ and then visits $b_p$, and we thus obtain the configuration $P_1$. 

Suppose now that  $b_p = b_q$. In this case the outgoing edge of $a_j$ cannot be $(a_j,b_p)$, because such an edge can only be added by $C_1$, which would imply $m_{j,p}=m_{j,q}=1$, violating the configuration above.  If  the path visits several vertices of $A$ and then visits $b_p(=b_q)$, then there must be a 1 
in the $q$th column below the $j$th row. We thus obtain the configuration $P_2$. 
\end{proof}

\subsubsection{Bipartite Graphs Representable for All Orderings: } 
The above forbidden submatrix characterization allows us to characterize the bipartite graphs $G=(A\cup B,E)$ that have a \Stick representation for \emph{every} possible ordering of $A$ and $B$.  Observe that the forbidden submatrices $P_2, P_3, P_1$ correspond, respectively, to the bipartite graphs shown in 
Fig.~\ref{fig:ccfab}(a)--(c).  
We can construct $2^2=4$ graphs from Fig.~\ref{fig:ccfab}(a) based on whether each of the dotted edges is present or not.
Similarly, we can construct $2^2=4$ graphs from Fig.~\ref{fig:ccfab}(b), and  $2^5=32$ graphs from Fig.~\ref{fig:ccfab}(b). Let $\mathcal{H}$ be the set that consists of these 40 graphs. 
From Theorem~\ref{thm:key} we immediately obtain:

\begin{figure}[pt]
\centering
\includegraphics[width=.5\textwidth]{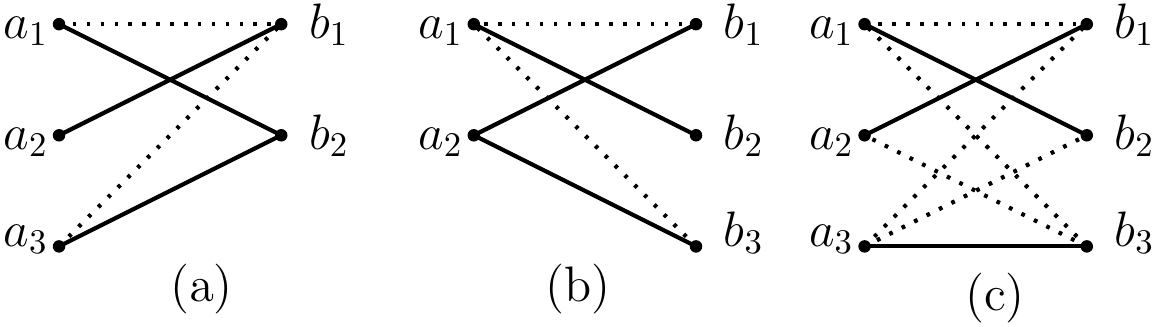}
  \caption{The forbidden subgraphs for Theorem~\ref{thm:ccfab}. Dotted edges are optional. 
  }
\label{fig:ccfab}
\end{figure}

\begin{theorem}\label{thm:ccfab}
A bipartite graph $G=(A\cup B,E)$ admits a \Stick representation for every possible ordering of $A$ and $B$ if and only if $G$ does not contain any graph of $\mathcal H$.
\end{theorem}


\section{Fixed As}
\label{sec:A-fixed}
 
In this section we study the \Stick representation problem when the ordering of only the vertices in $A$ is given. A formal description of the problem, which we call \StickA, is as follows.
\begin{tcolorbox}[colframe=gray,colback=lightgray!20!white,boxrule=1pt,arc=0.4em,boxsep=-1mm]
 \textbf{Problem:}  {\textsc{\Stick Representation with Fixed As (\StickA)}}\\
 \textbf{Input:} A bipartite graph $G = (A{\cup}B,E)$, and a vertex-ordering $\sigma_A$ of  $A$.\\
 \textbf{Question:}
 Does $G$ admit a \Stick representation such that the $i$th horizontal segment on the ground line corresponds to the $i$th vertex of $\sigma_A$?
\end{tcolorbox}

We give a polynomial-time algorithm for \StickA.
The idea is to use the forbidden submatrix characterization for \StickAB (Theorem~\ref{thm:key}).
We need an ordering of the $B$ vertices that, together with the given ordering $\sigma_A$, avoids the forbidden submatrices $P_1,P_2,P_3$.
We will express the conditions for the ordering of the $B$ vertices as a 2-SAT formula, i.e., a CNF (conjunctive normal form) formula where each clause contains at most two literals.
2-SAT can be solved in polynomial time~\cite{DBLP:journals/ipl/AspvallPT79}.

\begin{theorem}\label{thm:fa}
There is an algorithm with run-time $O(|A|^3|B|^3)$ to decide the \StickA problem, and construct a \Stick representation if one exists.
\end{theorem}
\begin{proof}
For each pair of vertices $v,w$ of $G$, we create variables $p_{v\prec w}$ and $p_{w\prec v}$ (representing the 
ordering of segments $v$ and $w$ on the ground line).
We will enforce $p_{v\prec w} = \neg p_{w\prec v}$ by adding clauses 
$(\neg{p_{v\prec w}} \vee \neg{p_{w\prec v}}) \wedge ({p_{v\prec w}} \vee {p_{w\prec v}})$. (One variable would suffice, but it is notationally easier to have both.)
We first set the truth values of all the variables involving two vertices of $A$ based on $\sigma_A$. We then add a few other clauses based on $P_1,P_2,P_3$, as follows.

For every $b_p,b_q,b_r$ giving rise to $P_1$, we add the clauses  $(\neg p_{{b_q\prec b_r}}\vee p_{b_q\prec b_p})$ and $(\neg p_{{b_p\prec b_q}}\vee p_{b_r\prec b_q})$. The first clause means that if $b_q\prec b_r$, then to avoid $P_1$, we must have $b_q\prec b_p$. Similarly, the second clause means if $b_p\prec b_q$, then to avoid $P_1$, we must have $b_r\prec b_q$. These clauses ensure  that if the SAT formula has a 
solution, then no configuration of the form $P_1$ can arise.

For every $b_p,b_q$ giving rise to $P_2$, we set $p_{b_q\prec b_p}$ to true. 
This would avoid any forbidden configuration of the form $P_2$ in a 
solution of the 2-SAT formula.

Finally, for every $b_p,b_q,b_r$ giving rise to $P_3$,  we add the  clauses $(\neg p_{{b_q\prec b_r}}\vee p_{b_q\prec b_p})$ and $(\neg p_{{b_p\prec b_q}}\vee p_{b_r\prec b_q})$. Note that these clauses can be interpreted in the same way as for $P_1$, i.e., if the 2-SAT formula has a 
solution, then no configuration of the form $P_3$ can arise.

Let $F$ be the resulting 2-SAT formula, which can be solved in linear time in the input size~\cite{DBLP:journals/ipl/AspvallPT79}, i.e., $O((|A| + |B|)^2)$ time.
If $F$ does not have a solution, then there does not exist any ordering of the $B$s that avoids the forbidden patterns. Thus $G$ does not admit the required \Stick representation. If $F$ has a solution, then there exists an ordering $\sigma_B$ of $B$s that together with $\sigma_A$ avoids all the forbidden patterns. By Theorem~\ref{thm:key}, $G$ admits the required \Stick representation, and it can be constructed from $\sigma_A$ and  $\sigma_B$ using Theorem~\ref{thm:fafb}.

Thus the time complexity of the algorithm is dominated by the time to construct the 2-SAT formula, which is $O(|A|^3|B|^3)$. Pseudocode for the Algorithm~\ref{algo2} is
given in Appendix~\ref{app:algo}.
\end{proof}

\smallskip\noindent
\textbf{Bipartite Graphs Representable for All $A$ Orderings: }
We also considered the class of bipartite graphs $G=(A \cup B, E)$ such that for every ordering of the vertices of $A$ there exists a \Stick representation.  We will call this the \StickfA class. Although we do not have a characterization of the \StickfA class, we describe some positive and negative instances below in Remark~\ref{r:a} and Remark~\ref{rem:b}, with proofs in
Appendix~\ref{app:fa}.

\begin{remark}\label{r:a}
Any bipartite graph $G=(A\cup B,E)$ with at most three vertices in $A$ belongs to the \StickfA class.
\end{remark}

\begin{remark}\label{rem:b}
A graph does not belong to \StickfA if its
bipartite adjacency matrix contains the submatrix 
\mbox{\scriptsize \begin{blockarray}{ccc}
\begin{block}{c[cc]}
  {$a_1$\ } & 1 & $*$ \\
  {$a_2$\ } & 0 & 1 \\
  {$a_3$\ } & 1 & 0 \\
  {$a_4$\ } & $*$ & 1 \\
\end{block}
\end{blockarray}} . %
(Here the columns are unordered.)
\end{remark}


\section{Stick graphs}\label{sec:abfree}

In this section we examine general \Stick representations, i.e., we do not impose any constraints on the ordering of the vertices. 

\begin{tcolorbox}[colframe=gray,colback=lightgray!20!white,boxrule=1pt,arc=0.4em,boxsep=-1mm]
 \textbf{Problem:}  {\textsc{\Stick Representation}}\\ 
 \textbf{Input:} A bipartite graph $G=(A\cup B,E)$.\\
 \textbf{Question:}
 Does $G$ admit a \Stick representation  such that the vertices in $A$ and $B$ correspond to horizontal and vertical segments, respectively?
\end{tcolorbox}

It is an open question to find a polynomial time algorithm for the above problem of recognizing \Stick graphs.
In this section we give some partial results (Remarks~\ref{r:y}--\ref{r:z}). 
The proofs of these remarks are included in Appendix~\ref{app:last}.

We begin with some definitions.
A matrix has the \emph{simultaneous consecutive ones} property if the rows and columns can be permuted so that the 1's in each row and each column  appear consecutively~\cite{DBLP:journals/tcs/OswaldR09}. A  \emph{one-sided drawing} of a planar bipartite graph $G=(A\cup B,E)$ is a planar straight-line drawing of $G$, where all vertices in $A$ lie on the $x$-axis, and the vertices of $B$ lie strictly above the $x$-axis~\cite{DBLP:conf/ciac/FossmeierK97}.

\begin{remark}\label{r:x}
Let $G=(A\cup B, E)$ be a bipartite graph and let $M$ be its adjacency matrix, where the rows and columns correspond to $A$s and $B$s, respectively. If $M$ has the simultaneous consecutive ones property, then $G$ admits a \Stick representation, which can be computed in $O(|A||B|)$ time. 
\end{remark}

\begin{proof}

One can determine whether $M$ has the simultaneous consecutive ones property in $O(|A||B|)$ time~\cite{DBLP:journals/tcs/OswaldR09}, and if so, then one can construct such a matrix $M'$ within the same time complexity. 

We now show how to construct the \Stick representation from $M'$. For each row (resp.,   column), we draw a horizontal (resp., vertical) segment starting from the rightmost (resp., topmost) 1 entry. We extend the horizontal segments to the left and vertical segments downward such that  they touch a ground line $\ell$, e.g., see Fig.~\ref{fig:rx}(a)--(b).

Let the resulting drawing be $D$, which may contain many unnecessary crossings. However, for each unnecessary crossing, we can follow the  segments involved in the crossings upward and rightward to find two distinct 1 entries.  Since the matrix has the simultaneous ones property, the violated entries in each row (column) must lie consecutively at the left end of the row (bottom end of the column). Therefore, one can find a $(+x,-y)$-monotone path $P$ that separates the violated entries from the rest of the matrix, e.g., see the shaded region in Fig.~\ref{fig:rx}(b).

Let $b_1,b_2,\ldots,b_k$ be the bend points creating $90^\circ$ angles towards $\ell$. To compute the required \Stick representation, we  remove these bends one after another, as follows. Consider the topmost bend point $b_i$, e.g., see $b_1$ in Fig.~\ref{fig:rx}(c). Imagine a Cartesian coordinate system with origin at $b_i$. Move the rows above $b_i$ and columns to the right of $b_i$ towards the upward and rightward directions, respectively, as illustrated in Fig.~\ref{fig:rx}(d). It is straightforward to observe that one now can construct a ground line $\ell'$ through $b_i$ such that the violated entries lie in the region below the path determined by $b_{i+1},\ldots,b_k$.
\end{proof}

\begin{remark}\label{r:y}
Let $G=(A\cup B, E)$ be an $n$-vertex bipartite graph that admits a one-sided planar drawing.  Then $G$ is a \Stick graph, and its \Stick representation can be computed in $O(n^2)$ time. 
\end{remark}

\begin{remark}\label{r:z}
Let $H$ be the graph obtained by deleting a perfect matching from  a complete bipartite graph $K_{4,4}$.  Any graph $G=(A \cup B, E)$ containing $H$ as an induced subgraph does not admit a \Stick representation. Since $H$ is a planar graph, not all planar bipartite graphs are \Stick graphs.
\end{remark}


\section{Open Problems}\label{sec:Conclusions}

\bigskip
\noindent
{\bf Open Problem 1.}
What is the complexity of recognizing \Stick graphs?
Is the problem NP-complete? 
By Theorem~\ref{thm:key} the problem is 
equivalent to ordering the rows and columns of a 0-1 matrix to exclude the 3 forbidden submatrices given in the Theorem statement. 
Note that 
these forbidden submatrices involve 5 or 6 rows and columns (vertices of the graph) so the 
results of Hell et al.~\cite{hell2014ordering}, which apply to patterns of at most 4 vertices in a bipartite graph, 
do not provide a polynomial time algorithm.

One possible approach using 3-SAT is as follows.
Given a bipartite graph $G=(A\cup B, E)$, one can create a 3-SAT formula $\Phi$ such that $\Phi$ is satisfiable if and only if $G$ admits a \Stick  representation, as follows. For each pair of vertices $v,w$ of $G$, create variables $p_{v\prec w}$ and $p_{w\prec v}$ (representing the ordering of $v$ and $w$ on the ground line), and add clauses $(\neg {p_{v\prec w}} \vee \neg {p_{w\prec v}}) \wedge ({p_{v\prec w}} \vee {p_{w\prec v}})$ to enforce $p_{v\prec w} = \neg p_{w\prec v}$. Now express the conditions C1 and C2 from Section~\ref{sec:AB-fixed} as 3-SAT clauses.

\begin{description}
\item[]\textbf{$\Phi_1$:} 
(Condition C1.) If $m_{i,p}=1$, then set $p_{a_i\prec b_p} = 1$.  
 
\item[]\textbf{$\Phi_2$:} 
(Condition C2.) We must express the condition that if the ordered submatrix 
\mbox{\scriptsize \begin{blockarray}{ccc}
& $b_p$ & $b_q$ \\
\begin{block}{c[cc]}
  {$a_i$\ } & 1 & $*$ \\
  {$a_j$\ } & 0 & 1 \\
\end{block}
\end{blockarray}} exists, then $p_{b_p\prec a_j} =1$.  
Thus, if $m_{i,p}=1, m_{j q}=1$ and $m_{j,p}=0$, then we add the clause $(\neg {p_{a_i\prec a_j}}\vee \neg {p_{b_p\prec b_q}} \vee \neg {p_{a_j\prec b_p}})$.

\item[]\textbf{$\Phi_3$:} For each triple $u,v,w$ of vertices, add the clause $(\neg {p_{u\prec v}} \vee \neg {p_{v\prec w}} \vee p_{u\prec w})$. Intuitively, these  are transitivity constraints, which would ensure a total ordering  on the ground line.
\end{description}

It is not difficult to show that the 3-SAT $\Phi$ is satisfiable if and only if $G$ admits the required intersection representation. However, since $\Phi$ contains $O(n^2)$ variables, using known SAT-solvers would not be faster than a naive algorithm that simply guesses the order of the segments along the ground line.
Therefore, an interesting direction for future research would be to find a 3-SAT formulation with a linear number of variables.

\bigskip\noindent 
{\bf Open Problem 2.} Can we improve the time complexity of the recognition algorithm for graphs with fixed $A$s?

\subsubsection{Acknowledgments}
The research of A. Lubiw and D. Mondal is supported in part by the Natural Sciences and Engineering Research Council of Canada (NSERC).
Also, this work is supported in part by NSF grants CCF-1423411 and CCF-1712119.

 \bibliographystyle{splncs04}
 
 \newpage
\appendix

\section{Fixed $A$s}
\label{app:fa}

\begin{proof}[Remark~\ref{r:a}]

The proof is by construction. For any ordering $\sigma_A$, we can categorize the columns into at most eight categories (assuming $A$ has exactly three vertices, the other cases are straightforward). We then organize those categories in the following order:
{\arraycontrol \begin{blockarray}{ccccccccc}
& $b_1$ & $b_2$ & $b_3$  & $b_4$ & $b_5$ & $b_6$ & $b_7$  & $b_8$\\
\begin{block}{c[cccccccc]}
  {$a_1$\ } & 1 & 0 & 1 & 0 & 0 & 1 & 1 & 0\\ 
  {$a_2$\ } & 0 & 1 & 1 & 0 & 1 & 1 & 0 & 0\\
  {$a_3$\ } & 0 & 0 & 0 & 1 & 1 & 1 & 1 & 0\\
\end{block}
\end{blockarray}} .
It is now straightforward to see that the resulting ordering corresponds to the required \Stick representation, e.g., see Fig.~\ref{fig:unf10011001}(a).
\end{proof}

 \begin{figure}[ht]
\centering
\includegraphics[width=.7\textwidth]{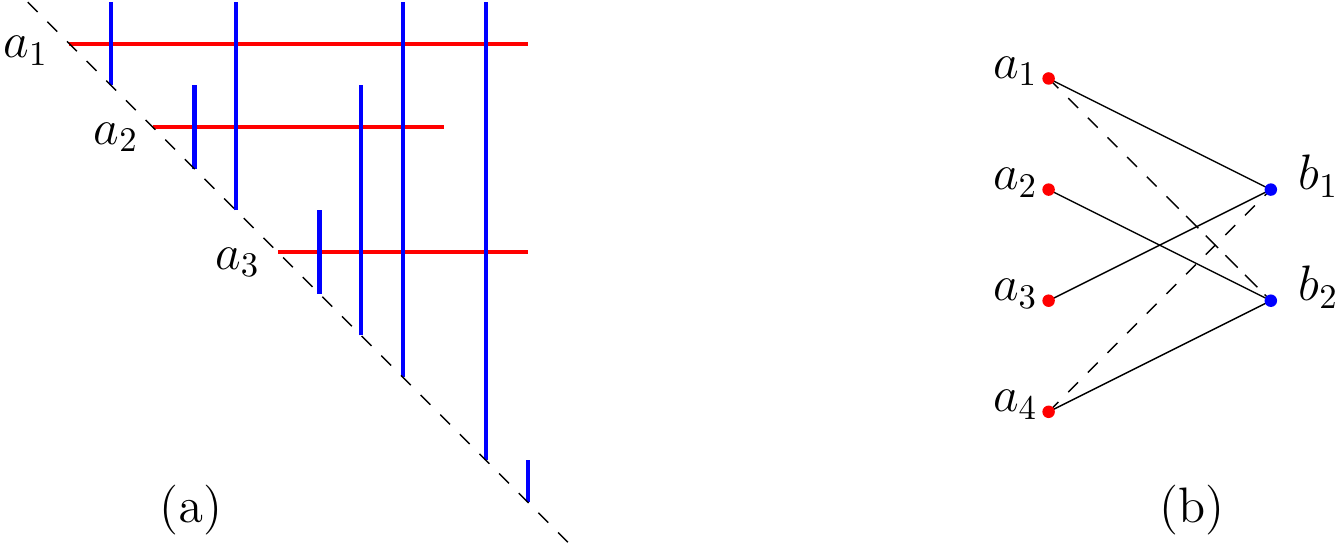}
\caption{(a) Illustration for Remark~\ref{r:a}. (b) Illustration for Remark~\ref{rem:b}.}
\label{fig:unf10011001}
\end{figure}

\begin{proof}[Remark~\ref{rem:b}]

We refer the reader to  Fig.~\ref{fig:unf10011001}(b). Any $\sigma_A$ consistent with the given matrix 
would serve as a  negative instance, as follows. The adjacency matrix 
{\arraycontrol \begin{blockarray}{cccccccc}
& $b_1$ & $b_2$  \\
\begin{block}{c[ccccccc]}
  {$a_1$\ } & 1 & *\\ 
  {$a_2$\ } & 0 & 1\\
  {$a_3$\ } & 1 & 0\\
\end{block}
\end{blockarray}}  %
would require $b_2$ to come after $b_1$, while the matrix 
{\arraycontrol \begin{blockarray}{cccccccc}
& $b_1$ & $b_2$  \\
\begin{block}{c[ccccccc]}
  {$a_2$\ } & 0 & 1\\ 
  {$a_3$\ } & 1 & 0\\
  {$a_4$\ } & * & 1\\
\end{block}
\end{blockarray}}  %
would require $b_1$ to come after $b_2$. Thus there is no consistent ordering for the vertices of $B$, and hence $G$ does not have a \Stick representation respecting $\sigma_A$. 
\end{proof}

\clearpage
\newpage

\section{Stick graphs}
\label{app:last}

\begin{figure}[ht]
\centering
\includegraphics[height=\textwidth]{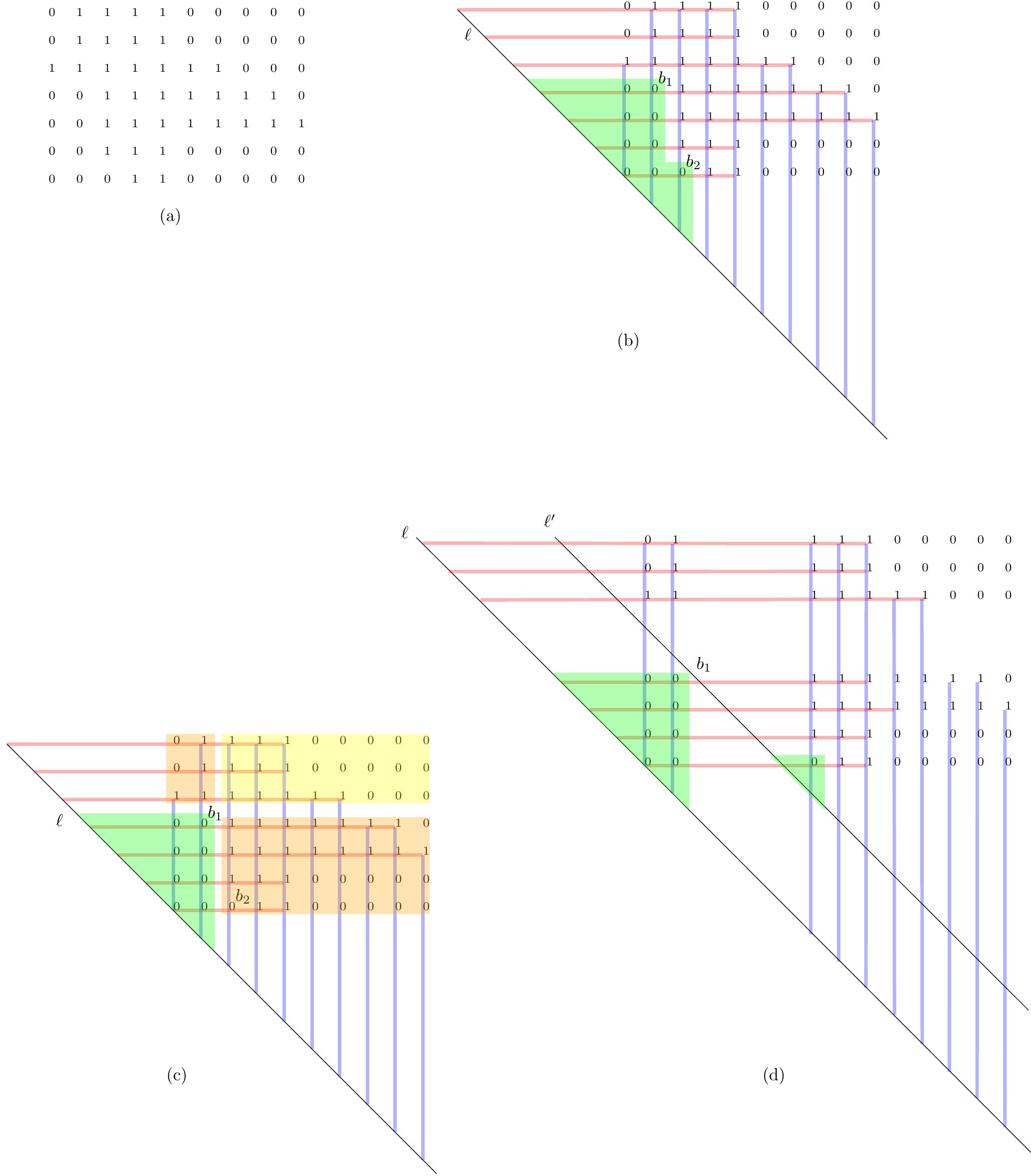}
\caption{Illustration for Remark~\ref{r:x}.}
\label{fig:rx}
\end{figure}

\begin{figure}[ht]
\centering
\includegraphics[height=\textwidth]{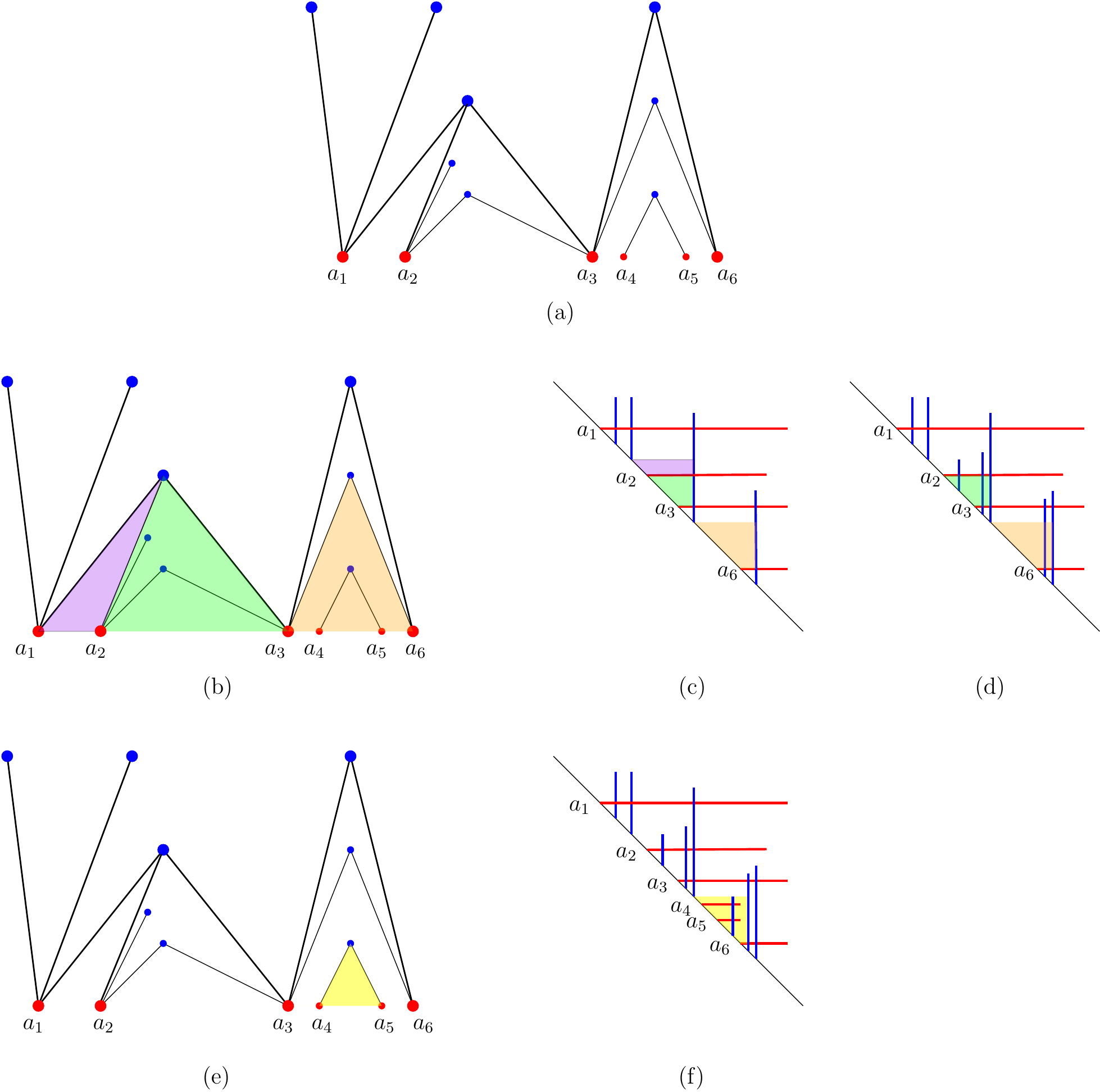}
\caption{Illustration for Remark~\ref{r:y}.}
\label{fig:ry}
\end{figure}

\begin{proof}[Remark~\ref{r:y}]
Given an $n$-vertex bipartite graph $G=(A\cup B, E)$, one can decide whether $G$ admits a  one-sided planar drawing in $O(n^2)$ time~\cite{DBLP:conf/ciac/FossmeierK97}, and if so, then one can construct such a drawing within the same time complexity. 

It now suffices to show how to construct a \Stick representation from the one-sided drawing. Let $\Gamma$ be a one-sided drawing of $G$. Imagine that we connected the $A$-vertices in a path from left to right in order they appear on the $x$-axis such that all the edges of the path are on the outerface. We define the  \emph{upper set} of $\Gamma$ to be the $B$ vertices that lie on the outerface, including  their neighbors.  

We start by drawing the upper set.
 It is straightforward to draw a \Stick representation by first drawing the $A$-vertices in the order they appear on the upper set from top to bottom on the ground line, and then the $B$-vertices to complete the \Stick representation of the set. Fig.~\ref{fig:ry}(a) illustrates a  one-sided drawing, and Fig.~\ref{fig:ry}(c) illustrates a drawing of its upper set.

Let $\Gamma'$ be the one-sided drawing by deleting the $B$-vertices of degree one from the upper set of $\Gamma$. Note that every $B$-vertex $w$ of degree two or more on the upper set, `covers'   smaller one-sided drawings below the edges that connect $w$ to a pair of consecutive $A$-neighbors, e.g., see the shaded regions  of Fig.~\ref{fig:ry}(b). We continue drawing each of these smaller one-sided drawings recursively between their corresponding  leftmost and rightmost $A$-vertices, which have already been drawn in the previous level. Fig.~\ref{fig:ry}(d) illustrates the drawing in the second recursion level, and Fig.~\ref{fig:ry}(e)--(f) illustrates the third recursion level.
\end{proof}

\begin{figure}[ht]
\centering
\includegraphics[height=.4\textwidth]{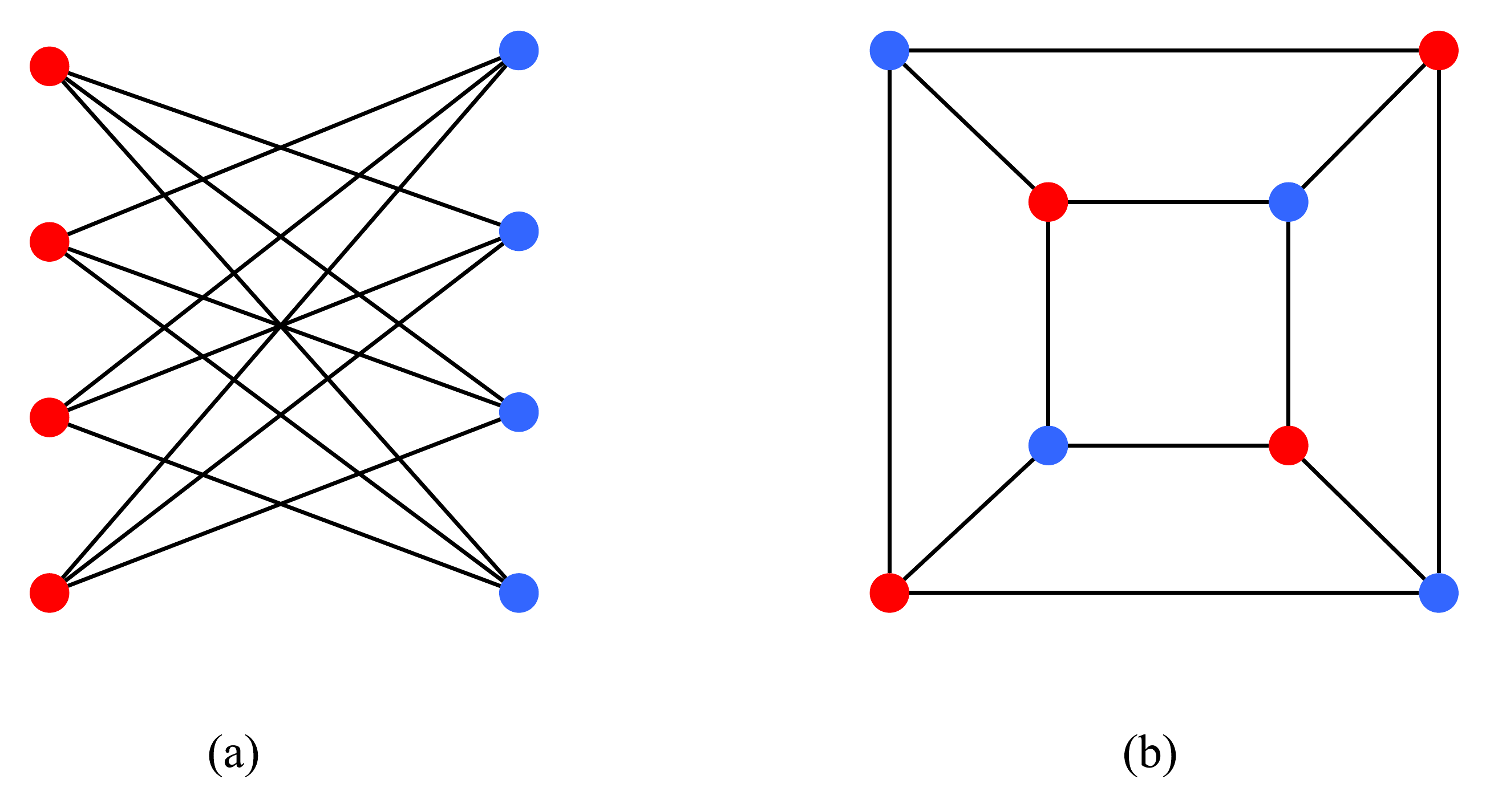}
\caption{Graph that does not admit a \Stick representation. Its matrix representation can have all zeros placed on a diagonal.}
\label{fig:unfzerodiag}
\end{figure}

\begin{proof}[Remark~\ref{r:z}]

It suffices to show that $H$  (Fig.~\ref{fig:unfzerodiag}(a)) does not admit a \Stick representation.
 It is straightforward to permute the rows and columns of the matrix representation of $H$ such that all zeros lie on the main diagonal, e.g.,
$M =${\arraycontrol \begin{blockarray}{cccc}
\begin{block}{[cccc]}
0 & 1 & 1 & 1\\ 
1 & 0 & 1 & 1\\
1 & 1 & 0 & 1\\
1 & 1 & 1 & 0\\
\end{block}
\end{blockarray}} .
Therefore, for every permutation  of rows and columns, we can pick the two columns that have zeros in the middle two rows to obtain the  configuration 
 {\arraycontrol \begin{blockarray}{cc}
\begin{block}{[cc]}
1 & *\\
0 & 1\\
1 & 0\\
* & 1\\
\end{block}
\end{blockarray}} .
 As we discussed in Sec.~\ref{sec:A-fixed}, there cannot be any \Stick representation with this row ordering. Consequently, the graph $H$, as well as  no graph containing $H$ as an induced subgraph, can have a \Stick representation.
 
Since $H$ is a planar graph (see Fig.~\ref{fig:unfzerodiag}(b)),  not all planar bipartite graphs are \Stick graphs.
 \end{proof}

 \section{Algorithms}
\label{app:algo}

\begin{algorithm}[htb]

\caption{Algorithm for Fixed-$A$-Fixed-$B$-\Stick Recognition}
\label{algo1}
\begin{algorithmic}[1]
\State{{\bf Input:} A bipartite graph $G=(A{\cup}B, E)$, an ordering $\sigma_A$ and an ordering $\sigma_B$}

\State $M \gets$ A matrix representation of $G$ whose rows follow the ordering $\sigma_A$, and columns follow the ordering $\sigma_b$

\State $H \gets $ A graph with vertex set $(A{\cup}B)$ but without any edge.

\For{\textbf{each } consecutive vertices $a_{i-1}, a_i$ in $\sigma_A$}
add the edge $(a_{i-1}, a_i)$ to $H$
\EndFor{}
  
\For{\textbf{each } consecutive vertices $b_{j-1}, b_j$ in $\sigma_B$}
add the edge $(b_{j-1}, b_j)$ to $H$
\EndFor{}

\For{\textbf{each } \text{entry } $m_{i,j}=1$ in $M$} add the edge $(a_i, b_j)$  
 to $H$
\EndFor{}

\For{\textbf{each } $a_i,a_j,b_p,b_q$ (respecting $\sigma_A$ and $\sigma_B$) of the form 
{\arraycontrol \begin{blockarray}{ccc}
& $b_p$ & $b_q$ \\
\begin{block}{c[cc]}
  {$a_i$\ } & 1 & $*$ \\
  {$a_j$\ } & 0 & 1 \\
\end{block}
\end{blockarray}} 
} add the edge $(b_p, a_j)$ to $H$
\EndFor{}

\If{$H$ contains a cycle}
\Return{false} //No solution exists
\EndIf{}

\State $\sigma \gets$ A topological ordering of the vertices of $H$\\ 
\Return{$\sigma$} 
\end{algorithmic}
\end{algorithm}

\begin{algorithm}[h]
\caption{Algorithm for Fixed-$A$-\Stick Recognition}
\label{algo2}
\begin{algorithmic}[1]

\State{{\bf Input:} A bipartite graph $G=(A \cup B, E)$, an ordering $\sigma_A$}

\State $P_1 \gets$ {\arraycontrol \begin{blockarray}{cccc}
& $b_p$ & $b_q$ & $b_r$\\
\begin{block}{c[ccc]}
  {$a_i$\ } & $*$  & $1$ & $*$ \\
  {$a_j$\ } & $*$  & $0$ & $1$ \\
  {$a_k$\ } & $1$  & $*$ & $*$ \\
\end{block}
\end{blockarray}} , %
$P_2 \gets${\arraycontrol \begin{blockarray}{ccc}
& $b_p$ & $b_q$ \\
\begin{block}{c[cc]}
  {$a_i$\ } & $1$ & $*$ \\
  {$a_j$\ } & $0$ & $1$ \\
  {$a_k$\ } & $1$ & $*$ \\
\end{block}
\end{blockarray}} , %
$P_3 \gets${\arraycontrol \begin{blockarray}{cccc}
& $b_p$ & $b_q$ & $b_r$\\
\begin{block}{c[ccc]}
  {$a_i$\ } & $*$ & $1$ & $*$ \\
  {$a_j$\ } & $1$ & $0$ & $1$ \\
\end{block}
\end{blockarray}} .

\State $\Phi\gets $ A 2-SAT with variables of the form $p_{u\prec w}$ and $p_{w\prec u}$, for every pair of vertices $u,w$ of $G$.

\For{\textbf{each } $a_i,a_j$ (respecting $\sigma_A$)}
$p_{a_i\prec a_j} \gets  {\it true}$,
$p_{a_j\prec a_i} \gets {\it false}$
\EndFor{} 

\For{{\textbf{each}} $P_1$ or $P_3$}
add the clauses $(\neg {p_{{b_q}\prec {b_r}}} \vee p_{{b_q}\prec {b_p}})$ and $(\neg {p_{{b_p}\prec {b_q}}} \vee p_{{b_r}\prec {b_q}})$
\EndFor{}

\For{{\textbf{each}} $P_2$}
$p_{b_q\prec b_p} \gets {\it true}$,
$p_{b_p\prec b_q} \gets {\it false}$
\EndFor{}

\For{{\textbf{each}} $P_3$}
add the clauses $(\neg {p_{{b_q}\prec {b_r}}} \vee p_{{b_q}\prec {b_p}})$ and $(\neg {p_{{b_p}\prec {b_q}}} \vee p_{{b_r}\prec {b_q}})$
\EndFor{}

\If{$\Phi$ does not admit an affirmative answer}
\Return{false} //No solution exists
\EndIf{}

\State $\sigma \gets$ A total ordering of vertices  determined  by any solution of $\Phi$ \\
\Return{$\sigma$} 

\end{algorithmic}
\end{algorithm}

\end{document}